\documentclass[12pt]{article}
\usepackage{amsmath,amsthm,amssymb}
\usepackage{array, longtable}
\newtheorem{Theorem}{Theorem}[section]
\newtheorem{lem}[Theorem]{Lemma}
\newtheorem{Remark}[Theorem]{Remark}
\newtheorem{Definition}[Theorem]{Definition}
\newtheorem{Corollary}[Theorem]{Corollary}
\newtheorem{Proposition}[Theorem]{Proposition}
\newtheorem{Example}[Theorem]{Example}

\numberwithin{equation}{section}
\numberwithin{table}{section}

\begin{document}

\title{Galois Hulls of Linear Codes over Finite Fields\footnote{
 E-Mail addresses: hwliu@mail.ccnu.edu.cn (H. Liu), panxu@mails.ccnu.edu.cn (X. Pan)}}

\author{Hongwei Liu,~Xu Pan}
\date{\small
School of Mathematics and Statistics, Central China Normal University\\Wuhan, Hubei, 430079, China\\
}
\maketitle

\begin{abstract}
The $\ell$-Galois hull  $h_{\ell}(C)$ of an $[n,k]$ linear code $C$ over a finite field $\mathbb{F}_q$ is the intersection of  $C$ and  $C^{{\bot}_{\ell}}$, where  $C^{\bot_{\ell}}$ denotes the $\ell$-Galois dual of $C$ which introduced by Fan and Zhang (2017).  The $\ell$- Galois LCD code is a linear code $C$ with $h_{\ell}(C) = 0$. In this paper, we show that the dimension of the $\ell$-Galois hull of a linear code is invariant under permutation equivalence and we provide a method to calculate the dimension of the $\ell$-Galois hull  by the generator matrix of the code. Moreover, we obtain that the dimension of the $\ell$-Galois hulls of ternary codes are also invariant under monomial equivalence.
We show that every $[n,k]$ linear code over $\mathbb F_{q}$ is monomial equivalent to an $\ell$-Galois LCD code for any $q>4$. We conclude that if there exists an $[n,k]$ linear code over $\mathbb F_{q}$ for any $q>4$, then there exists an $\ell$-Galois LCD code with the same parameters for any $0\le \ell\le  e-1$, where $q=p^e$ for some prime $p$. As an application, we characterize  the $\ell$-Galois hull of matrix product codes over finite fields.


\end{abstract}
\section{Introduction}
Let $\mathbb{F}_q$ be a finite field of order $q$, where $q=p^e$ and $p$ is a prime.  Recently, Fan and Zhang \cite{FZ} generalize the Euclidean inner product and the Hermitian inner product to the so-called {\it $\ell$-Galois form} (or {\it $\ell$-Galois inner product}), where $0\le \ell\le e-1$. The $\ell$-Galois dual codes, and the $\ell$-Galois self-dual constacyclic codes over finite fields are studied. In particular, necessary and sufficient conditions for the existence of $\ell$-Galois self-dual  and  isometrically Galois self-dual constacyclic codes are obtained. As consequences, some results on self-dual, iso-dual and Hermitian self-dual constacyclic codes are derived.

Linear complementary-dual (LCD for short) codes are linear codes that intersect with their duals
trivially. They were first  studied by Massey \cite{M} who showed that these codes are optimal for the two-user binary adder channel (BAC) and
that they are asymptotically good. Sendrier \cite{S3} showed that these codes meet the Gilbert-Varshamov bound. In (\cite{S2},\cite{S3},\cite{S4},\cite{SSS}, \cite{SS}), the authors also studied the hulls of linear codes, and tried to find permutations between two equivalent codes, which has an application to code-based public key cryptosystems. Carlet and Guilley gave some applications of LCD codes in side-channel attacks and fault non-invasive attacks (\cite{BC},\cite{CG1},\cite{CG2}). LCD codes also can  be used for constructions of lattices \cite{J}. Optimal and MDS codes that are LCD are studied in many papers (see~\cite{BJ},\cite{CM},\cite{CL},\cite{FZ},\cite{HO},\cite{L},\cite{LL},\cite{M},\cite{S1},\cite{KS}).

Motivated by the previous work, we study the Galois hulls of linear codes over finite fields. The {\it $\ell$-Galois hull} of a linear code $C$ over a finite field $\mathbb{F}_q$ is defined by $h_{\ell}(C) = C\bigcap C^{\bot_{\ell}}$, where $q=p^e$, $p$ is a prime, and $0\le \ell\le e-1$. The classical LCD code is a linear code with $h_{0}(C) = 0$, and the Hermitian LCD code is a code with  $h_{e\over 2}(C)=0$, where $e$ is even.

Construction of codes is an interesting research field in coding theory. The matrix product code $[C_{1}, \cdots, C_{M}]\cdot A$ is a new code constructed from the codes $C_1,\cdots, C_M$ of  same length $n$  and an $M\times N$ matrix $A$ over a finite field $\mathbb{F}_q$. These codes were first proposed and studied in \cite{BN}. There are many papers focusing on its algebraic structure, different distance structures, and decoding algorithm  (see~\cite{HR},\cite{HL},\cite{MS},\cite{A},\cite{FL}).

This paper is organized as follows. Section 2 gives some preliminaries. In Section~3, a characterization  of  the dimension of the $\ell$-Galois hull of a linear code is provided. As a corollary, we obtain a necessary and sufficient condition for a linear code to be an $\ell$-Galois LCD code. In Section~4, we first show that the dimension of any  $\ell$-Galois hull of a linear code is invariant under permutation equivalence for $0\le \ell \le e-1$.  For  ternary codes, the dimension of the $\ell$-Galois hull is also invariant under monomial equivalence. Then we show that every  linear code over $\mathbb F_{q}$ is monomial equivalent to an $\ell$-Galois LCD code in the case of $q>4$. We conclude that if there exists an $[n, k]$ linear code  over $\mathbb F_{q}$ with $q>4$, then there exists an $\ell$-Galois LCD code with the same parameters.  In Section~5,  we  study the structure and the dimension of $\ell$-Galois hull of matrix product codes.

\section{Preliminaries}
Throughout this paper, $\mathbb F_{q}$ denotes a finite field of order $q=p^{e}$, where $p$ is a prime, $e$ is a positive integer. By $\mathbb{F}_{q}^{*}$ we denote the multiplicative group of $\mathbb{F}_{q}$. Let  $\mathbb{F}_{q}^{n} = \{ (x_{1},\cdots,x_{n})\ |\ x_{j} \in \mathbb{F}_{q}, 1\le j\le n\}$ be the $n$ dimensional vector space over $\mathbb{F}_{q}$. Any subspace $C$ of $\mathbb{F}_{q}^{n}$ is called a {\it linear code} of length $n$ over $\mathbb{F}_{q}$. We assume that all codes  are linear in this paper.

Let $S_{n}$ be the {\it symmetric group} on the set $X=\{1,2,\cdots ,n\}$. For all $\varphi \in S_{n}$ and  $\mathbf{x}=(x_{1},\cdots,x_{n}) \in \mathbb{F}_{q}^{n}$, $S_{n}$  acts on $\mathbb{F}_{q}^{n}$ in the following way.
$$
S_n\times \mathbb{F}_{q}^{n}\to \mathbb{F}_{q}^{n}, (\varphi, \mathbf{x})\mapsto \varphi\cdot\mathbf{x}=\varphi(\mathbf{x})=(x_{\varphi(1)},\cdots,x_{\varphi(n)}).
$$

In \cite{FZ}, Fan and Zhang introduced the following concept.

\begin{Definition}
Assume the notations given above. For each integer $\ell$ with $0\leq \ell\leq e-1$, let
$$\langle\mathbf{x},\mathbf{y}\rangle_{\ell}=x_{1}y_{1}^{p^{\ell}}+ \cdots +x_{n}y_{n}^{p^{\ell}} , \quad \forall\,  \mathbf{x},\mathbf{y}\in \mathbb{F}_{q}^{n}. $$
Then the form $\langle-,-\rangle_{\ell}$ is called the  $\ell$-Galois form on $\mathbb{F}_{q}^{n}$, or $\ell$-Galois inner product.
\end{Definition}

It is easy to see that $\langle-,-\rangle_{0}$ is just the usual Euclidean inner product. And,  $\langle-,-\rangle_{e\over 2}$ is the Hermitian inner product if $e$ is even. For any code $C$ over $\mathbb{F}_{q}$, the following code
$$
C^{\perp _{\ell} }=\{ \mathbf{x} \in \mathbb{F}_{q}^{n}\, |\, \langle\mathbf{c},\mathbf{x}\rangle_{\ell}=0,\,\, \forall \mathbf{c} \in C \}
$$
is called the {\it $\ell$-Galois dual code} of $C$. If $C\subseteq {C^{\perp _{\ell} }}$, then $C$ is said to be {\it $\ell$-Galois self-orthogonal}. Moreover, $C$ is said to be {\it $\ell$-Galois self-dual} if $C= {C^{\perp _{\ell} }}$.

Note that $C^{\perp _{\ell} }$ is linear whenever $C$ is linear or not. In particular,  $C^{\perp _{0} }$ ($C^{\perp }$ for short) is just the Euclidean dual code of $C$, and $C^{\perp _{ \frac{e}{2}} }$($C^{\perp _{H} }$ for short) is just the Hermitian dual code of $C$ if $e$ is even.

Let $\sigma:\mathbb{F}_{q}\rightarrow\mathbb{F}_{q}, a\mapsto a^{p}$, be the Frobenius automorphism of $\mathbb{F}_{q}$. For any $\mathbf{x}=(x_{1},\cdots,x_{n}) \in \mathbb{F}_{q}^{n}$, and any matrix $G=(a_{ij})_{k\times l}$ over $\mathbb{F}_{q}$ , set $\sigma(\mathbf{x}):=(\sigma(x_{1}),\cdots,\sigma(x_{n}))$ and $\sigma(G):=(\sigma(a_{ij}))_{k\times l }$.

The following proposition is easily obtained.

\begin{Proposition}
Assume the notations given above. Then for any $0\le \ell\le e-1$,

(1) $C^{\perp _{\ell} }=(\sigma^{e-\ell}(C))^{\perp_{0}}=\sigma^{e-\ell}(C^{\perp_{0}}).$

(2) $(C^{\perp _{\ell} })^{\perp _{f}}=\sigma^{2e-\ell-f}(C)$, for any $0\le \ell,f\le e-1$.  In particular, $(C^{\perp _{0} })^{\perp _{0}}=C$, and $(C^{\perp _{\frac{e}{2}} })^{\perp _{\frac{e}{2}}}=C$ if $e$ is even.
\end{Proposition}

\begin{proof}
The two statements  follow immediately from the identity $\langle\mathbf{c},\mathbf{x}\rangle_{\ell}=\langle\mathbf{c},\sigma^{\ell}(\mathbf{x})\rangle_{0}$$=\sigma^{\ell}(\langle\sigma^{e-\ell}(\mathbf{c}),\mathbf{x}\rangle_{0}).$
\end{proof}


\begin{Definition}
Let $C$ be a linear code over $\mathbb{F}_{q}$. The $\ell$-Galois hull of $C$ is defined by $h_{\ell}(C) = C\bigcap C^{\bot_{\ell}}$. If $h_{\ell}(C) = 0$, then $C$ is called a linear code with $\ell$-Galois complementary dual or an $\ell$-Galois LCD code. If $h_{\ell}(C) = C$, then $C$ is called an $\ell$-Galois self-orthogonal linear code.
\end{Definition}
\begin{Remark}
Note that when $\ell=0$  and $C\bigcap C^{\bot_0} = 0$,  the code $C$ is the classical LCD code. When  $e$ is even and $C\bigcap C^{\bot_{e\over 2}} = 0$,  the code $C$ is the Hermitian LCD code.
 \end{Remark}
A {\it monomial} matrix is a square matrix such that in every row (and in every column) there is exactly one nonzero element. It is easy to see that any monomial matrix is a product of a permutation matrix and an invertible diagonal matrix.  In particular, a permutation matrix is a special monomial matrix. Two linear codes $C_{1}$ and $C_{2}$ of length $n$ over $\mathbb{F}_q$ are {\it monomial equivalent}, if there is an monomial matrix $M$ of size $n$ such that $C_{2}=C_{1}M=\{\mathbf{y}\mid \mathbf{y}=\mathbf{x}M, \,\,for \, \mathbf{x}\in C_{1} \}$. If $M$ is a permutation matrix, then $C_{1}$ and $C_{2}$ are called {\it permutation equivalent}.

\section{The $\ell$-Galois hull of linear codes }
In this section, we give a characterization for the $\ell$-Galois hull of any linear code over $\mathbb{F}_q$. We have the following theorem.
\begin{Theorem}\label{l-Galois}
Let $C$ be an $[n,k]$ linear code over $\mathbb{F}_{q}$ with a generator matrix $G$.
Let $h$ be the dimension of the $\ell$-Galois hull $h_{\ell}(C)= C\bigcap C^{\bot_{\ell}}$ of $C$, and let $r = k-h$. Then there exists a generator
matrix $G_{0}$ of $C$ such that
\[
G_{0}\sigma^{\ell}(G_{0}^{T})=\left(\begin{array}{cc}
                        O_{h\times h} & H_{h\times r} \\
                        O_{r\times h} & P_{r\times r}
\end{array}\right),
\]
where $O_{h \times h}$ and $O_{r \times h}$ are respectively zero matrices of sizes $h\times h$ and $r\times h$,  and the rank $r(Q)$ of $Q=
\left(\begin{array}{c}
                         H_{h\times r} \\
                         P_{r\times r}
\end{array}\right)$ is $r$. Furthermore, the rank $r(G\sigma^{\ell} (G^{T}))$ of $G\sigma^{\ell} (G^{T})$ is $r$ for any generator matrix $G$ of $C$.

 \end{Theorem}

\begin{proof}
Let $\{{\bf \alpha}_{1},\cdots, {\bf \alpha}_{h}\}$ be a basis of $h_{\ell}(C)$. We can extend $\{{\bf \alpha}_{1},\cdots, {\bf \alpha}_{h}\}$ to a basis $\{{\bf \alpha}_{1}, \cdots, \alpha_{h},\cdots, \alpha_{k}\}$ of $C$. Let $G_{0}$ be the $k\times n$ matrix such that its $i$th row is $\alpha_{i}$, where $1\le i\le k$. Then $G_{0}$ is a generator matrix of $C$ and $G_{0}\sigma^{\ell} (G_{0}^{T})$ is a $k \times k$ matrix. The element at the $(i,j)$-entry of $G_{0}\sigma^{\ell} (G_{0}^{T})$ is ${\bf \alpha}_{i}\sigma^{\ell}({\bf \alpha}_{j}^{T})=\langle{\bf \alpha}_{i},{\bf \alpha}_{j}\rangle_{\ell}$. Note that $\langle\alpha_{i},\alpha_{j}\rangle_{\ell}=0$ if $1\leq j\leq h$, since $\alpha_{i} \in C$ for all $1\leq i \leq k$ and $\alpha_{j} \in C^{\bot_{\ell}} $ for all $1\leq j\leq h$. Therefore, $G_{0}\sigma^{\ell}(G_{0}^{T})$ has the form as stated in the theorem.

Now we show that $r(Q)=r$. Obviously, $r(Q)\le r$. Suppose $r(Q)< r$. Then there exists a non-zero vector $ \tilde{{\bf x}}=(x_{h+1},\cdots,x_{k})\in \mathbb{F}_{q}^{k-h}$ such that $Q(\tilde{{\bf x}})^{T}=0$. Let $\mathbf x =({\bf 0},\tilde{{\bf x}}) \in \mathbb{F}_{q}^{k},$
where ${\bf 0}$ is the zero vector of length $h$. Then we have
$$
G_{0}\sigma^{\ell}(G_{0}^{T})\mathbf x^{T}=(0, Q)\left(\begin{array}{c}{\bf 0} \\  \tilde{{\bf x}}^T\end{array}\right)=0.
$$

Since the map $\sigma^{\ell}:\mathbb{F}_{q}\rightarrow \mathbb{F}_{q},\sigma^{\ell}(a)=a^{p^{\ell}},\,\, \forall a \in \mathbb{F}_{q} $ is an automorphism of $\mathbb{F}_q$, there exists a vector $\mathbf y =(0,\cdots,0,y_{h+1},\cdots,y_{k}) \in \mathbb{F}_{q}^{k}$ such that $\sigma ^{\ell}(\mathbf y)=\mathbf x$. Therefore,
$$
0=G_{0}\sigma^{\ell}(G_{0}^{T})\mathbf x^{T}=G_{0}\sigma^{\ell}(G_{0}^{T})\sigma ^{\ell}(\mathbf y)^{T}=G_{0}\sigma^{\ell}(G_{0}^{T})\sigma ^{\ell}(\mathbf y^{T})=G_{0}\sigma^{\ell}(G_{0}^{T}\mathbf y^{T})=G_{0}\sigma^{\ell}(\mathbf y G_{0})^{T}.
$$
This gives that ${\mathbf y}G_{0} \in  C^{\bot_{\ell}}$, which implies $ {\mathbf y}G_{0} \in h_{\ell}(C)=\langle\alpha_{1},\cdots, \alpha_{h}\rangle$. We also have
$$
{\bf y}G_{0}=y_{h+1} \alpha_{h+1}+y_{h+2} \alpha_{h+2}+\cdots +y_{k} \alpha_{k} \in \langle\alpha_{h+1}, \alpha_{h+2},\cdots, \alpha_{k}\rangle.
$$
Hence ${\mathbf y}={\bf 0}$ since $\alpha_{1}, \alpha_{2},\cdots, \alpha_{k} $ are linear independent. This is a contradiction. Hence,  $r(Q)=r$.

Let  $G$ be an arbitrary  generator matrix of $C$, then there exists an invertible $k \times k$ matrix $N$ such that $G=NG_{0}$. We have
$$
G\sigma^{\ell}(G^{T})=NG_{0} \sigma^{\ell} (G^{T}_{0}N^{T})=NG_{0} \sigma^{\ell} (G^{T}_{0}) \sigma^{\ell}(N^{T}).
$$
Then $r(G\sigma^{\ell}(G^{T}))=r(G_{0}\sigma^{\ell}(G_{0}^{T}))$, since the matrix $N$ and $\sigma^{\ell}(N^{T})$ are invertible. We are done.
\end{proof}

The following corollary can be obtained immediately.

 \begin{Corollary}(\cite{Ruud-P})
  Let $C$ be an $[n,k]$ linear code over $\mathbb{F}_{q}$ with a generator matrix $G$.
Let $h$ be the dimension of $h_{0}(C)$ and $r = k-h$. Then the code $C$ has a generator
matrix $G_{0}$ such that
\[
G_{0}G_{0}^{T}=\left(\begin{array}{cc}
                        O_{h\times h} & O_{h\times r} \\
                        O_{r\times h} & P_{r\times r}
\end{array}\right),
\]
where $O_{h\times h}, O_{h\times r}, O_{r\times h}$ are all zero matrices, and $P$ is an invertible $r\times r$ matrix. Furthermore, the rank of $GG^{T}$ is $r$ for every generator matrix $G$ of $C$.
 \end{Corollary}

\begin{proof}
Take $\ell=0$ in Theorem~\ref{l-Galois}, we get $G_{0}G_{0}^{T}=G_{0}\sigma^{0}(G_{0}^{T})=\left(\begin{array}{cc}
                        O_{h\times h} & H_{h\times r} \\
                        O_{r\times h} & P_{r\times r}
\end{array}\right)$. Note that $G_{0}G_{0}^{T}$ is a symmetric matrix, this implies that $H_{h \times r}=O_{h \times r}$. By Theorem~\ref{l-Galois} again, the rank of $
\left(\begin{array}{c}
                         H_{h\times r} \\
                         P_{r\times r}
\end{array}\right)$ is $r$, we have $P$ is an invertible $r\times r$ matrix. Also the rank of $GG^T=G\sigma^{0} (G^{T})$ is $r$ for every generator matrix $G$ of $C$ by Theorem~\ref{l-Galois}.
\end{proof}

\begin{Corollary}
  Let $C$ be an $[n,k]$ linear code over $\mathbb{F}_{q}$ with a generator matrix $G$, where  $q=p^{e}$, and $e$ is even.
Let $h$ be the dimension of $h_{\frac{e}{2}}(C)$ and $r = k-h$. Then $C$ has a generator
matrix $G_{0}$ such that
\[
G_{0}\sigma^{e\over 2}({G}_{0}^{T})=\left(\begin{array}{cc}
                        O_{h\times h} & O_{h\times r} \\
                        O_{r\times h} & P_{r\times r}
\end{array}\right),
\]
where $O_{h\times h}, O_{h\times r}, O_{r\times h}$ are all zero matrices, and $P$ is an invertible $r\times r$ matrix. Furthermore, the rank of $G\sigma^{e\over 2}({G}^{T})$ is $r$ for every generator matrix $G$ of $C$.
 \end{Corollary}

\begin{proof}
Take $\ell={e\over 2}$ in Theorem~\ref{l-Galois}, and note that $(G_{0}\sigma^{e\over 2}(G_{0}^{T}))^T=\sigma^{e\over 2}(G_{0})G_{0}^{T}$. It is easy to verify that  $\langle\mathbf{x},\mathbf{y}\rangle_{\frac{e}{2}}=\sigma^{\frac{e}{2}}(\langle\mathbf{y},\mathbf{x}\rangle_{\frac{e}{2}})$ for any ${\bf x}, {\bf y}\in {\mathbb F}_q^n$. Hence  $\langle\mathbf{x},\mathbf{y}\rangle_{\frac{e}{2}}=0$ if and only if $\langle\mathbf{y},\mathbf{x}\rangle_{\frac{e}{2}}=0$. The result then follows immediately.
\end{proof}

\begin{Remark}
The following example shows that if $\ell\neq 0$, or  $\ell \ne \frac{e}{2}$, where $e$ is even, then the matrix $H_{h\times r}$ may not be $0$. For example, let $\mathbb F_{8}=\mathbb F_{2}[w] (w^{3}+w+1=0)$ and $C\le \mathbb F_{8}^{4}$ be a $[4, 2]$ linear code of length $4$ with a generator matrix $G=\left(\begin{array}{cccc}
                        1 &1 & w+1 & w+1 \\
                        0 & w^{2}+1 & 1 &0
\end{array}\right)$.
Then $G\sigma(G^{T})=\left(\begin{array}{cc}
                       0 & w^{2} \\
                        0 &w^{2}+w+1
\end{array}\right) $.
\end{Remark}

When $h_l(C)=0$ or $h_l(C)=C$ for a linear code $C$, then the following two corollaries  are straightforward.

\begin{Corollary}(\cite{LFL})\label{cor-3-1}
Let $C$ be an $[n,k]$ linear code over $\mathbb{F}_{q}$ with a generator matrix $G$. Then $C$ is $\ell$-Galois LCD code if and only if $G\sigma^{\ell} (G^{T})$ is nonsingular.
\end{Corollary}

\begin{Corollary}
Let $C$ be an $[n,k]$ linear code over $\mathbb{F}_{q}$ with a generator matrix $G$. Then~$C$ is an $\ell$-Galois  self-orthogonal code if and only if $G\sigma^{\ell} (G^{T})=0$.
\end{Corollary}

In particular, we have
\begin{Corollary}
Let $C$ be an $[n,k]$ linear code over $\mathbb{F}_{q}$ with a generator matrix $G$ and a parity check matrix $H$. Then $C$ is an $\ell$-Galois self-dual code if and only if both $G\sigma^{\ell} (G^{T})$ and $H\sigma^{e-\ell} (H^{T})$ are $0$.
\end{Corollary}
\begin{proof}
Since  $H$ is a parity check matrix of $C$, $\sigma^{e-\ell}(H)$ is a generator matrix of $C^{\perp_{\ell}}$. Note that $C\subseteq C^{\perp_{\ell}}$ if and only if $G\sigma^{\ell}(G^T)=0$, and $C^{\perp_{\ell}}\subseteq C=(C^{\perp_{\ell}})^{\perp_{e-\ell}}$ if and only if $\sigma^{e-\ell}(H)\sigma^{e-\ell}(\sigma^{e-\ell}(H))^T=0$ if and only if $H\sigma^{e-\ell}(H^T)=0$. This finishes the proof.
\end{proof}

\section{The existence of $\ell$-Galois LCD codes }

LCD codes over finite fields are an important class of linear codes. They have many applications in coding theory and cryptography, especially in designing decoding algorithm. In this section, we focus on the equivalence of $l$-Galois LCD codes.

\begin{lem} Let $C$ be an $[n,k]$ linear code of length $n$ over $\mathbb{F}_{q} $, $\varphi \in S_{n}$ be a permutation, and $0\le l\le e-1$. Then

(1) $\langle\varphi(\mathbf{x}),\varphi(\mathbf{y})\rangle_{\ell}=\langle\mathbf{x},\mathbf{y}\rangle_{\ell}$, for any $\mathbf{x}, \mathbf{y}\in \mathbb{F}_{q}^{n}$.

(2) $\varphi(C^{\perp _{\ell}})=\varphi(C)^{\perp _{\ell}}$.
\end{lem}

\begin{proof}
(1) Let $\mathbf{x}=(x_{1},\cdots,x_{n}), \mathbf{y}     =(y_{1},\cdots,y_{n}) \in \mathbb{F}_{q}^{n}$, then $\varphi(\mathbf{x})     =(x_{\varphi(1)},\cdots,x_{\varphi(n)}) $ and $\varphi(\mathbf{y})     =(y_{\varphi(1)},\cdots,y_{\varphi(n)}) $. Hence
$$
\langle\varphi(\mathbf{x}),\varphi(\mathbf{y})\rangle_{\ell}=x_{\varphi(1)}y_{\varphi(1)}^{p^{\ell}}+ \cdots +x_{\varphi(n)}y_{\varphi(n)}^{p^{\ell}}=x_{1}y_{1}^{p^{\ell}}+ \cdots +x_{n}y_{n}^{p^{\ell}}=\langle\mathbf{x},\mathbf{y}\rangle_{\ell}.
$$

(2) Let $\mathbf{v}\in C^{\perp_{\ell}}$, then $\varphi(\mathbf{v})\in\varphi( C^{\perp_{\ell}})$. For any $\mathbf{u}\in\varphi(C)$,  there exists a codeword $\mathbf{c}\in C$ such that $\varphi(\mathbf{c})=\mathbf{u}$. We have $$\langle\mathbf{u},\varphi(\mathbf{v})\rangle_{\ell}=\langle\varphi(\mathbf{c}),\varphi(\mathbf{v})\rangle_{\ell}=\langle\mathbf{c},\mathbf{v}\rangle_{\ell}=0.$$
This implies that  $\varphi(\mathbf{v})\in\varphi( C)^{\perp_{\ell}}$ and hence $\varphi(C^{\perp _{\ell}})\subseteq\varphi(C)^{\perp _{\ell}}$ . Since
 $$\dim(\varphi(C^{\perp _{\ell}}))=\dim(C^{\perp _{\ell}})=n-\dim(C)=n-\dim(\varphi(C))=\dim(\varphi(C)^{\perp _{\ell}}),$$
we get $\varphi(C^{\perp _{\ell}})=\varphi(C)^{\perp _{\ell}}$.
\end{proof}

\begin{Proposition}\label{permutation-equivalent}
The dimension of the $\ell$-Galois hull of a linear code is invariant under permutation equivalence.
\end{Proposition}

\begin{proof}
Let $C$ be an $[n,k]$ linear code of length $n$ over $\mathbb{F}_{q} $ and  $\varphi \in S_{n}$ be an arbitrary permutation. By Lemma~4.1, we have
$$\varphi(h_{\ell}(C))=\varphi(C\cap C^{\perp _{\ell}})=\varphi(C)\cap \varphi(C^{\perp _{\ell}})=\varphi(C)\cap \varphi(C)^{\perp _{\ell}}=h_{\ell}(\varphi(C)).$$
Because $\dim(h_{\ell}(\varphi(C))=\dim(\varphi(h_{\ell}(C)))=\dim(h_{\ell}(C))$, this finishes the proof.
\end{proof}
\begin{Remark}
If  $q=3$, the dimension of the $\ell$-Galois hull of a linear code is invariant under monomial equivalence. In fact, suppose that $C_{2}=C_{1}M$,  by Proposition~\ref{permutation-equivalent}, we only need to  prove the  case when $M$ is an invertible diagonal matrix. Let $G_{1}$ and $G_{2}$ be the generator matrices of $C_{1}$ and $C_{2}$ respectively, then $G_{2}=G_{1}M$.  Since $q=3$,  we have $\ell=0$ and $MM^{T}=I$, the identity matrix. Hence we have
$$G_{2}\sigma^{\ell}(G_{2})^{T}=G_{1}M\sigma^{\ell}(G_{1}M)^{T}=G_{1}M\sigma^{\ell}(M)^{T}\sigma^{\ell}(G_{1})^{T}=G_{1}MM^{T}\sigma^{\ell}(G_{1})^{T}=G_{1}\sigma^{\ell}(G_{1})^{T}.$$
It follows that  $h_{\ell}(C_{1})=h_{\ell}(C_{2})$.
\end{Remark}

In order to prove the main result in this section, we need the following proposition. This proposition is known (for example, see~\cite{Ruud-P}), we provide an alternative proof here.
\begin{Proposition}\label{proposition-1}
 Let $f(X)$ be a nonzero polynomial of $\mathbb{F}_{q}[X_{1},\cdots ,X_{n}]$ such
that the degree of $f(X)$ with respect to $X_{j}$ is at most $q-1$ for all $j$, where  $1\le j\le n$. Then
there exists a vector $\mathbf x\in \mathbb{F}_{q}^{n}$ such that $f(\mathbf x)\neq 0$.
\end{Proposition}

\begin{proof}
We prove this proposition by induction on $n$. If $n=1$, then by assumption the degree of $f(X)$ is at most $q-1$. Note that the number of  roots of $f(X)$ over the finite field $\mathbb{F}_q$ is less or equal to the degree of $f(X)$. Therefore, there exists a vector $\mathbf x\in \mathbb{F}^1_{q} $ such that $f(\mathbf{x})\neq 0$.

Now assume $n\geq 2$. We can further assume that the degree $\deg_{X_{n}}(f)$ of $f(X)$ with respect to $X_n$ is greater than or equal to $1$. Otherwise $f(X)\in \mathbb{F}_{q}[X_{1},\cdots ,X_{n-1}]\subseteq \mathbb{F}_{q}[X_{1},\cdots ,X_{n}]$, then there exists a vector $\mathbf x\in \mathbb{F}_{q}^{n-1}$ such that $f(\mathbf x)\neq 0$ by the inductive hypothesis. Therefore we can assume that
$$f(X)=f(X_{1},\cdots,X_{n})=\sum_{i=1}^{k}g_{i}(X_{1},\cdots,X_{n-1})X_{n}^{i},$$
where  $k=\deg_{X_{n}}f(X_{1},\cdots,X_{n})$ and $g_{k}(X_{1},\cdots,X_{n-1})$ is a nonzero polynomial. Hence, there exists $\bar{\mathbf x}=(x_{1},\cdots,x_{n-1})\in \mathbb{F}_{q}^{n-1} $ such that $g_{k}(x_{1},\cdots,x_{n-1})\neq 0$. Let $F(X_{n})=f(x_{1},\cdots,x_{n-1},X_{n})$ and $\deg(F(X_{n}))=k$, then there is an element $x_{n}\in \mathbb{F}_{q}$ such that $F(x_{n})\neq 0$ by the result of the previous argument. Hence there exists $\mathbf x=(x_{1},\cdots,x_{n})$ such that $f(\mathbf x)\neq 0$.
\end{proof}

\begin{Proposition}
 Let $f(X)$ and $g(X)$ be two nonzero polynomials of $\mathbb{F}_{q}[X_{1},\cdots ,X_{n}]$ such
that the degree of $f( X)g( X)$ with respect to $X_{j}$ is at most $q-1$ for all $j$. Let $\Omega=\{ \mathbf x\in \mathbb{F}_{q}^{n}\, |\, g(\mathbf x)= 0\}$. Then
there exists a vector $\mathbf x\in \mathbb{F}_{q}^{n}\backslash \Omega $ such that $f(\mathbf x)\neq 0$.
\end{Proposition}

\begin{proof}
By Proposition~\ref{proposition-1}, there exists a vector $\mathbf x\in \mathbb{F}_{q}^{n}$ such that $f(\mathbf x)g(\mathbf x)\neq 0$.  Hence there exists a vector $\mathbf x\in \mathbb{F}_{q}^{n}\backslash \Omega $ such that $f(\mathbf x)\neq 0$.
\end{proof}

 By the proposition above, we can easily get the following two corollaries.

\begin{Corollary}\label{cor-1}
 Let $f(X)$ be a nonzero polynomial of $\mathbb{F}_{q}[X_{1},\cdots ,X_{n}]$ such
that the degree of $f(X)$ with respect to $X_{j}$ is at most $q-2$ for all $j$. Then
there exists an $\mathbf x\in \{\mathbb{F}^*_{q}\}^{n}$ such that $f(\mathbf x)\neq 0$.
\end{Corollary}
\begin{proof}
Let $g(X)=\prod_{i=1}^{n}X_{i}$ and $\Omega=\{ \mathbf x\in \mathbb{F}_{q}^{n} \,|\,g(\mathbf x)= 0\}$. Then $\mathbb{F}_{q}^{n}\backslash\Omega=\{\mathbb{F}^*_{q}\}^{n}$. We are done.
\end{proof}

\begin{Corollary}\label{4-6}
 Let $q=p^{e},1\leqslant \ell\leqslant e-1$ and $\ell |e$.  Let $f(X)$ be a nonzero polynomial of $\mathbb{F}_{q}[X_{1},\cdots ,X_{n}]$ such
that the degree of $f(X)$ with respect to $X_{j}$ is at most $q-1-p^{\ell}$ for all $j$. Then
there exists a vector $\mathbf x\in \{\mathbb{F}_{q}\setminus \mathbb{F}_{p^{\ell}} \}^{n}$ such that $f(\mathbf x)\neq 0$.
\end{Corollary}

\begin{proof}
Let $g(X)=\prod_{i=1}^{n}(X_{i}^{p^{\ell}}-X_{i})$ and $\Omega=\{ \mathbf x\in \mathbb{F}_{q}^{n}\, |\, g(\mathbf x)= 0\}$. Then $\mathbb{F}_{q}^{n}\backslash\Omega=(\mathbb{F}_{q}\backslash\mathbb{F}_{p^{\ell}})^{n}$. The result then follows immediately.
\end{proof}

\begin{Theorem}\label{4-7}
Let $C$ be an $[n,k]$ linear code over $\mathbb{F}_{q}$, where $q>4$ and $0\leqslant \ell\leqslant e-1$. Then $C$ is monomial equivalent to an $\ell$-Galois LCD code.
\end{Theorem}

\begin{proof}
Let $C$ be an $[n,k]$ linear code over $\mathbb{F}_{q}$ with $q>4$. Without loss of generality, we may assume that $C$ has a generator matrix of the standard form $G=(I_{k} \mid B)$. Let $\mathbf{X}     =(X_{1},\cdots,X_{k}) $. Now we define a $k$-variable polynomial $f(\mathbf{X})$ as follows: $$f(\mathbf{X})=\det(diag(X^{1+p^{\ell}}_{1},\cdots,X^{1+p^{\ell}}_{k})+B\sigma^{p^{l}}(B^{T})).$$

It is easy to verify that $f(\mathbf{X})$ is a nonzero polynomial with  the variables $X_{1},\cdots,X_{k}$ and the degree of $f(\mathbf{X})$ with respect to $X_{j}$ is $1+p^{\ell}$ for all $j$.  In the following, we show that $1+p^{\ell}\leq q-2$. Note that $q>4$, we have

If $p=2$ and $e\geq3$, then $2^{e-1}\geq3$, and $2^{e-1}+2^{e-1}\geq 3+2^{e-1}$. This implies that   $2^{e}\geq3+2^{e-1}\geq 3+2^{l}$. Therefore, $q-2=2^{e}-2\geq 2^{\ell}+1$.

If $p\geq 3$ and $e\geq 2$, then $p^{e}\geq 3p^{e-1}\geq p^{e-1}+3$, hence $q-2=p^{e}-2\geq p^{e-1}+1\geq p^{\ell}+1$.

If $p\geq 5$ and $e=1$, then $\ell=0$, and we get $q-2>2=1+p^{\ell}$. Hence, $1+p^{\ell}\leq q-2$.

Therefore, by Corollary~\ref{cor-1}, there exists a vector $\mathbf{x}=(x_{1},\cdots,x_{k}) \in \{\mathbb{F}^*_{q}\}^{k}$ such that $f(\mathbf{x})\neq 0$.

Now let $G_{\mathbf{x}}=GM$ be the $k\times n$ matrix, where $M$ is an $n\times n$ diagonal matrix with the form $M=diag(x_{1},\cdots,x_{k},1,\cdots,1)$. Let $C_{\mathbf{x}}=CM$ be the code with the generator matrix $G_{\mathbf{x}}$. Then $C_{\mathbf{x}}$ is monomial equivalent to $C$. And we have $$\det(G_{\mathbf{x}}\sigma^{l}(G_{\mathbf{x}})^{T})=\det(diag(x^{1+p^{\ell}}_{1},\cdots,x^{1+p^{\ell}}_{k})+B\sigma^{p^{\ell}}(B^{T}))=f(\mathbf{x})\neq 0.$$
Therefore, $C_{\mathbf{x}}$ is an $\ell$-Galois LCD code by Corollary~3.5.
\end{proof}

When $\ell=0$, we have the following corollary.
\begin{Corollary}(\cite{R}, \cite{Ruud-P})
Let $C$ be a linear code over $\mathbb{F}_{q}$ with $q>4$. Then $C$ is monomial equivalent to an LCD code.
\end{Corollary}
In fact, this corollary is also true when $q=4$ (see \cite{R}). When $e$ is even, and $\ell={e\over 2}$, we get the following corollary.

\begin{Corollary}(\cite{R}, \cite{Ruud-P})
Let $C$ be a linear code over $\mathbb{F}_{q}$ with $q>4$ and $q$ a square. Then $C$ is monomial equivalent to an Hermitian LCD code.
\end{Corollary}

\begin{Remark}
When $q=4$, the corollary above is not right in general. In fact, if we let $\mathbb{F}_{4}=\{0,1,w,\bar{w}\}$ be the finite field of order $4$, where $\bar{w}=w^{2}=w+1$. Let $C=\langle(1,1)\rangle=\{(1,1),(0,0),(w,w),(\bar{w},\bar{w})\}$, and let $G=(1,1)$ be a generator matrix of~$C$. Then $G\bar{G}^{T}=0$ and so $C$ is not an Hermitian LCD code. For any $(a,b)\in {\mathbb{F}_{q}^{*}}^{2}$, let $G_{(a,b)}=G\left(\begin{array}{cc}
                       a  &0   \\
                        0  & b

\end{array}\right)$ be a generater matrix of the code $C_{(a,b)}$. Then $$G_{(a,b)}\bar{G}_{(a,b)}^{T}=a\bar{a}+b\bar{b}=a\bar{a}(1+a^{-1}b\bar{a^{-1}}\bar{b})=a\bar{a}(1+c\bar{c})(\mbox{let}\, c=a^{-1}b).$$
Then for all $c=1,w,\bar{w}$, we have $G_{(a,b)}\bar{G}_{(a,b)}^{T}=0$. Hence any code that is monomial equivalent with $C$ is not an Hermitian LCD code.
\end{Remark}

\begin{Theorem}\label{quick-LCD}
Let $C$ be a linear code over $\mathbb{F}_{q}$, where $q=p^{e},0\leqslant \ell\leqslant e-1,1\leqslant m\leqslant e-1,m|e$ and $p^{e}-p^{\ell}-p^{m}\geq 2$. Then there exists a $\mathbf{x} =(x_{1},\cdots,x_{k}) \in\{\mathbb{F}_{q}\backslash\mathbb{F}_{p^{m}}\}^{k}$ such that $C_{\mathbf{x}}$ is an $\ell$-Galois LCD code.
\end{Theorem}

\begin{proof}
Let $C$ be an $[n,k]$ linear code over $\mathbb{F}_{q}$. Without loss of generality, we may assume that $C$ has a generator matrix of the form $G=(I_{k} \mid B)$. Let $\mathbf{x}     =(x_{1},\cdots,x_{k}) \in \{\mathbb{F}_{q}\backslash\mathbb{F}_{p^{m}}\}^{k}$ and $G_{\mathbf{x}}$ be the the generator matrix of the code $C_{\mathbf{x}}$, where $C_\mathbf{x}$ is defined in Theorem~\ref{4-7}.  Let $\mathbf{X}     =(X_{1},\cdots,X_{k}) $. Now we define $f(\mathbf{X})=\det(diag(X^{1+p^{\ell}}_{1},\cdots,X^{1+p^{\ell}}_{k})+B\sigma(B)^{T})$. Hence $f(\mathbf{X})$ is a polynomial with the variables $X_{1},\cdots,X_{k}$ and the degree of $f(\mathbf{X})$ with respect to $X_{i}$ is $1+p^{\ell}$ for all $i$. We know that $2\leq p^{e}-p^{\ell}-p^{m}$ and $1+p^{\ell}\leq p^{e}-p^{m}-1$. The leading term of $f(\mathbf{X})$ with respect to the total degree of lex order is $X^{1+p^{\ell}}_{1},\cdots,X^{1+p^{\ell}}_{k}$. So $f(\mathbf{X})$ is a nonzero polynomial. Therefore $f(\mathbf{x})\neq 0$ for some $\mathbf{x} \in \{\mathbb{F}_{q}\backslash\mathbb{F}_{p^{\ell}}\}^{k}$ by Corollary~\ref{4-6}. Hence $C_{\mathbf{x}}$ is an $\ell$-Galois LCD code by this choice of $\mathbf{x}$, and Corollary~\ref{cor-3-1} because $\det(G_{\mathbf{x}}\sigma^{p^{\ell}}(G_{\mathbf{x}})^{T})=f(\mathbf{x})\neq 0$.
\end{proof}

\begin{Remark}
Sometimes we want to find $\mathbf{x}=(x_{1},\cdots,x_{k})\in \mathbb{F}_q^k$ and $C_{\mathbf{x}}$ for a linear code $C$ such that $C_{\mathbf{x}}$ is  an $\ell$-Galois LCD code. It is easy to find $\mathbf{x}=(x_{1},\cdots,x_{k})$ by using Theorem~\ref{quick-LCD} rather than Theorem~\ref{4-7}, because the set ${\mathbb{F}_{q}\backslash\mathbb{F}_{p^{m}}}$ is smaller than the set ${\mathbb{F}_{q}\backslash \{0\}}$.
\end{Remark}

\section{An application to matrix product codes}
In this section, we apply the results obtained in Section~4 to study the hull of matrix product codes over finite fields.

Let $A$ be an $M\times N$  matrix, $B$ be an $R\times S$ matrix. The tensor product of the two matrices is defined by $A\otimes B =\left(\begin{array}{ccc}
                       a_{11}B  &\cdots  &a_{1N}B \\
                        \vdots  & \ddots &\vdots \\
                        a_{M1}B & \cdots &a_{MN}B
\end{array}\right)$. The following properties of the tensor product of matrices are well-known.

\begin{lem}
Let $A\in \mathbb{F}_{q}^{M\times N}, B\in \mathbb{F}_{q}^{R\times S},C\in \mathbb{F}_{q}^{N\times T}$ and $D\in \mathbb{F}_{q}^{S\times U}$. Then

(1)\, $(A\otimes B)^{T}=A^{T}\otimes B^{T}$.

(2)\, $(A\otimes B)(C\otimes D)=(AC)\otimes (BD)$.

\end{lem}

\begin{Definition}
Let $A=(a_{ij})$ be an $M\times N$ matrix over $\mathbb{F}_{q}$, and let $C_{1},\cdots,C_{M}$ be codes of length $n$ over $\mathbb{F}_{q}$. The matrix product code $[C_{1},\cdots, C_{M}]\cdot A$ is the set of all matrix products $\big{[}{\bf c}_{1}, \cdots, {\bf c}_{M}\big{]}\cdot A = \big{[}\sum_{i=1}^{M}{\bf c}_{i}a_{i1},\cdots, \sum_{i=1}^{M}{\bf c}_{i}a_{iN}\big{]}$, where ${\bf c}_{i}\in C_{i}$ is a $1 \times n$ row vector for $i=1,\cdots,M$.
\end{Definition}
By using the tensor product of matrices, a matrix product code can be written as $[{\bf c}_{1}, \cdots, {\bf c}_{M}]\cdot A=[{\bf c}_{1}, \cdots, {\bf c}_{M}](A\otimes I)$, the usual matrix product, where $I$ is the $n\times n$ identity matrix.
\begin{Remark}
In the original paper (see \cite{BN}), the authors defined the matrix product code by writing each codeword ${\bf c}_j\in C_j$ as a column vector. Here we write each codeword ${\bf c}_j\in C_j$ as a row vector in the above definition, which is different from the original paper. However, by using the alternative definition, we can easily obtain the generator matrix of the matrix product code and  calculate its hull.
\end{Remark}
Recall that a {\it right inverse} of an $M\times N$ matrix $A$ is an $N\times M$ matrix $B$ such that $AB=I_{M}$. In this case we say that $A$ is {\em right non-singular}. Throughout  if $A$ is right non-singular then $A^{-1}$ denotes a right inverse of $A$. The following proposition on matrix product codes is well-known.

\begin{Proposition}(\cite{BN})\label{Prop-5-1}
Let $C_{1},\cdots,C_{M}$ be $M$ codes of length $n$ over $\mathbb{F}_{q} $. If the matrix $A$ is right non-singular then $\big{|}[C_{1},\cdots, C_{M}]\cdot A\big{|}=|C_{1}|\cdots |C_{M}|$. Furthermore, if $C_{1},\cdots, C_{M}$ are linear codes, then $\dim([C_{1},\cdots, C_{M}]\cdot A)=\sum_{i=1}^{M}\dim(C_{i})$.

\end{Proposition}


The following result is well-known (see \cite{BN}). We give an another proof.
\begin{Proposition}
Let $C_{1},\cdots,C_{M}$ be linear codes of length $n$ over $\mathbb{F}_{q}$ and the matrix $A$ be right non-singular. Let $C=[C_{1},\cdots, C_{M}]\cdot A$ with generator matrix $G$, and let $G_{i}$ be the generator matrix of $C_{i}$ for all $1\leq i\leq M$. Then $G=\left(\begin{array}{ccc}
                       a_{11}G_{1}  &\cdots  &a_{1N}G_{1} \\
                        \vdots  & \ddots &\vdots \\
                        a_{M1}G_{M} & \cdots &a_{MN}G_{M}
\end{array}\right)$.
\end{Proposition}

\begin{proof}
Let $G_{i}=\left(\begin{array}{ccc}
                      {\bf\alpha}_{i1}  \\
                        \vdots   \\
                        \alpha _{ik_{i}}
\end{array}\right)$, $\alpha_{ij} \in \mathbb{F}_{q}^{n}$ and $k_{i}$ be the dimension of $C_{i}$ for all $1\leq i\leq M$. For any ${\bf c}=[{\bf c}_{1,}\cdots, {\bf c}_{M}]\cdot A\in C$, ${\bf c}_{i}=[\lambda_{i1}\cdots\lambda_{ik_{i}}]G_{i}$. We have

\begin{equation*}
\begin{split}
{\bf c}&=[{\bf c}_{1},\cdots, {\bf c}_{M}](A\otimes I)=[\lambda_{11},\cdots, \lambda_{1k_{1}},\cdots, \lambda_{M1},\cdots ,\lambda_{Mk_{M}}]\left(\begin{array}{cccc}
                       G_{1}  &0  &\cdots & 0 \\
                        0  & G_{2} &\cdots&\vdots \\
                        \vdots & \vdots &\ddots &0\\
                        0 & \cdots & 0 &G_{M}
\end{array}\right)(A\otimes I)\\
&=[\lambda_{11},\cdots, \lambda_{1k_{1}},\cdots, \lambda_{M1},\cdots, \lambda_{Mk_{M}}]\left(\begin{array}{ccc}
                       a_{11}G_{1}  &\cdots  &a_{1N}G_{1} \\
                        \vdots  & \ddots &\vdots \\
                        a_{M1}G_{M} & \cdots &a_{MN}G_{M}
\end{array}\right).
\end{split}
\end{equation*}
Since the matrix $A$ is non-singular, the number of rows of $G$ is $\sum_{i=1}^{M}k_{i}=\dim(C)$ by Proposition~\ref{Prop-5-1}. Hence $G$ is a generator matrix of $C$.
\end{proof}

We have the following theorem.

\begin{Theorem}
Let $C_{1},\cdots,C_{M}$ be $[n, k_i]$ linear codes of length $n$ over $\mathbb{F}_{q}$, where $1\le i\le M$, and let $C=[C_{1},\cdots, C_{M}]\cdot A$ be the matrix product code.  For  $0\le \ell\le e-1$,  suppose $A\sigma^{\ell}(A^{T})=diag(\lambda_{1},\cdots,\lambda_{M})$, then the $\ell$-Galois hull of $C$  is
$$h_{\ell}(C)=[B_{1},\cdots, B_{M}]\cdot A,$$
where for all 1$\le i\le M$, $B_{i} = \left\{ \begin{array}{ll}
C_{i},  & \textrm{if $\lambda_{i}=0 ;$}\\
h_{\ell}(C_{i}),  & \textrm{if $\lambda_{i}\neq0 .$}
\end{array} \right.$
\end{Theorem}

\begin{proof}

Suppose that $\mathbf{v}=[{\bf v}_{1},\cdots, {\bf v}_{M}]\cdot A\in [B_{1},\cdots, B_{M}]\cdot A $, then ${\bf v}\in C=[C_{1},\cdots, C_{M}]\cdot~A$, because $B_{i}\subseteq C_{i}$ for all $1\leq i\leq M$. Let  $\mathbf{u}=[{\bf u}_{1},\cdots, {\bf u}_{M}]\cdot A\in C$, we have
\begin{equation*}
\begin{split}
\langle\mathbf{u},\mathbf{v}\rangle_{\ell}&=\mathbf{u}\sigma ^{\ell}(\mathbf{v}^{T})=[{\bf u}_{1},\cdots, {\bf u}_{M}](A\otimes I)\sigma^{\ell}(([{\bf v}_{1},\cdots, {\bf v}_{M}](A\otimes I))^{T})\\
&=[{\bf u}_{1}\cdots {\bf u}_{M}](A\otimes I)\sigma^{\ell}((A^{T}\otimes I)\left(\begin{array}{c} {\bf v}_{1}^{T} \\ \vdots    \\  {\bf v}_{M}^{T} \end{array}\right))\\
&=[{\bf u}_{1},\cdots, {\bf u}_{M}](A\otimes I)(\sigma^{\ell}(A^{T})\otimes I)\left(\begin{array}{c} \sigma^{\ell}({\bf v}_{1}^{T})   \\ \vdots    \\ \sigma^{\ell}({\bf v}_{M}^{T})\end{array}\right)
=[{\bf u}_{1}, \cdots, {\bf u}_{M}](A\sigma^{\ell}(A^{T})\otimes I)\left(\begin{array}{c} \sigma^{\ell}({\bf v}_{1}^{T})   \\ \vdots    \\ \sigma^{\ell}({\bf v}_{M}^{T})\end{array}\right)\\
&=[{\bf u}_{1},\cdots, {\bf u}_{M}](\left(\begin{array}{cccc} \lambda_{1}I  &0  &\cdots & 0 \\ 0  & \lambda_{2}I &\cdots&\vdots \\ \vdots & \vdots &\ddots &0\\ 0 & \cdots & 0 &\lambda_{M}I
\end{array}\right))\left(\begin{array}{c} \sigma^{\ell}({\bf v}_{1}^{T})   \\ \vdots    \\ \sigma^{\ell}({\bf v}_{M}^{T} \end{array}\right)\\
&=[\lambda_{1}{\bf u}_{1},\cdots, \lambda_{M}{\bf u}_{M}]\left(\begin{array}{c} \sigma^{\ell}({\bf v}_{1}^{T})   \\ \vdots    \\ \sigma^{\ell}({\bf v}_{M}^{T}) \end{array}\right)=\sum\limits_{i=1}^{M}\lambda_{i}{\bf u}_{i}\sigma^{\ell}({\bf v}_{i}^{T}).
\end{split}
\end{equation*}

If $\lambda_{i}\neq 0$, then ${\bf v}_{i}\in B_{i}=h_{\ell}(C_{i})$ and ${\bf u}_{i}\sigma^{\ell}({\bf v}_{i}^{T})=0$ for all $1\leq i\leq M$. Hence $\mathbf{v}\in h_{\ell}(C)$ and $[B_{1},\cdots, B_{M}]\cdot A\subseteq h_{\ell}(C)$.

Now suppose $\mathbf{v}=[{\bf v}_{1},\cdots, {\bf v}_{M}]\cdot A\in h_{\ell}(C) $. Assume $\lambda_{i}=0$, then ${\bf v}_{i}\in B_{i}=C_{i}$. If $\lambda_{i}\neq0$. For any ${\bf u}_{i}\in C_{i}$, let $\mathbf{u}_{i}=[0,\cdots,0,u_{i},0,\cdots,0]$ and the number of location of $u_{i}$ be $i$. Since $\mathbf{u}_{i}\in C$, we know that $0=\langle\mathbf{u}_{i},\mathbf{v}\rangle_{\ell}=\lambda_{i}{\bf u}_{i}\sigma^{\ell}({\bf v}_{i}^{T})$. Hence ${\bf u}_{i}\sigma^{\ell}({\bf v}_{i}^{T})=0$ and ${\bf v}_{i}\in h_{\ell}(C_{i})=B_{i}$. Thus, $\mathbf{v}\in[B_{1},\cdots, B_{M}]\cdot A$ and $h_{\ell}(C)\subseteq [B_{1},\cdots, B_{M}]\cdot A$.
\end{proof}

\begin{Corollary}\label{5-7}
Let $C_{1},\cdots,C_{M}$ be $[n, k_i]$ linear codes of length $n$ over $\mathbb{F}_{q}$, where $1\le i\le M$, and let $C=[C_{1},\cdots, C_{M}]\cdot A$ be the matrix product code.  For  $0\le \ell\le e-1$,  suppose $A\sigma^{\ell}(A^{T})=diag(\lambda_{1},\cdots,\lambda_{M})$, where $\lambda_i\ne 0$ for all $i$.  Then the $\ell$-Galois hull of $C$  is
$$
h_{\ell}(C)=[h_{\ell}(C_1),\cdots, h_{\ell}(C_{M})]\cdot A.
$$
\end{Corollary}
Denote $C^{\bot_{0}}$ by $C^{\bot}$ and $h_{0}(C)$ by $h(C)$. Then

\begin{Corollary}
Let $C_{1},\cdots,C_{M}$ be linear codes of length $n$ over $\mathbb{F}_{q}$, and $C=[C_{1},\cdots, C_{M}]\cdot~A$.  If the matrix $A$ satisfies   $AA^{T}=diag(\lambda_{1},\cdots,\lambda_{M})$ and $\lambda_{i}\neq 0$ for all $1\leq i\leq M$,  then the hull $h(C)$ of $C$ is $h(C)=[h(C_{1}),\cdots, h(C_{M})]\cdot A$.

\end{Corollary}

\begin{proof}
Take  $\ell=0$ in Corollary~\ref{5-7},  the result then follows.
\end{proof}

\begin{Example}
Take $q=3$ and $l=0$. Let  $C_{1}$ and $C_{2}$ be linear codes of length $4$ over $\mathbb F_{3}$ with generator matrices $G_{1}=\left(\begin{array}{cccc}  1  &0 &1&1\\  0   &1 &1&-1 \end{array}\right)$ and $G_{2}=\left(\begin{array}{cccc} 1    &1  &1  &1\end{array}\right) $  respectively. Then $h_{0}(C_{1})=0$ and $h_{0}(C_{2})=C_{2}$. Let $A=\left(\begin{array}{cc} 1    &1 \\ -1   &1
\end{array}\right)$ and $C=[C_{1},C_{2}]\cdot A$, it is easy to verify that  the generator matrix $G$ of $C$ is $G=\left(\begin{array}{cccccccc}
                       1    &0 &1&1&1 &0 &1 &1\\
                        0   &1 &1&-1 &0 &1 &1 &-1\\
                        -1  &-1 &-1 &-1&1&1  &1&1
\end{array}\right)$ and $GG^{T}=\left(\begin{array}{ccc}
                       0    &0  &0\\
                        0   &0  &0\\
                        0&  0&    -1
\end{array}\right)$.  Hence the generator matrix of $h_{0}(C)$ is $\left(\begin{array}{cccccccc}
                      1    &0 &1&1&1 &0 &1 &1\\
                        0   &1 &1&-1 &0 &1 &1 &-1\\
\end{array}\right) $  by Theorem~3.1. Since $h_{0}(C_{1})=0$, $h_{0}(C_{2})=C_{2}$,  and $AA^{T}=\left(\begin{array}{cc}
                       -1    &0 \\
                        0   &-1
\end{array}\right)$, we know that  the generator matrix of $h_{0}(C)=[h_{0}(C_{1}),h_{0}(C_{2})]\cdot A$ also is  $\left(\begin{array}{cccccccc}
                      1    &0 &1&1&1 &0 &1 &1\\
                        0   &1 &1&-1 &0 &1 &1 &-1\\
\end{array}\right) $ by Corollary~5.7.
\end{Example}

A matrix $A=(a_{ij})\in \mathbb{F}_{q}^{s\times s}$ is called {\it upper triangular} if $a_{ij}=0$ whenever $i>j$. A matrix $A\in\mathbb{F}_{q}^{s\times s}$ is called {\it block upper triangular} if as a block matrix $A$ is partitioned into the submatrices $A_{ij}\in \mathbb{F}_{q}^{s_{i}\times s_{j}}$, so that $A=(A_{ij})_{t\times t }$ , $\sum_{i=1}^{t}s_{i}=s$ and $A_{ij}=0$ for all $i>j$ , $1\leq i\leq s$ and $1\leq j\leq s$. Consider the $A_{ij}$ block as the entries of $A$, $A$ is upper triangular. Thus, $A=\left(\begin{array}{cccccc}
                       A_{11} &A_{12} &\cdots & A_{1t} \\
                        0  & A_{22}&\cdots&\vdots \\
                        \vdots &   &\ddots \\
                        0 & \cdots & 0 & A_{tt}
\end{array}\right),$ where each $A_{ij}$ is a $s_{i}\times s_{j}$ matrix and $\sum_{i=1}^{p}s_{i}=s$.

The following lemma is easy to prove.

\begin{lem}
Suppose  $A=\left(\begin{array}{cccccc}
                       A_{11} &A_{12} &\cdots & A_{1t} \\
                        0  & A_{22}&\cdots&\vdots \\
                        \vdots &   &\ddots \\
                        0 & \cdots & 0 & A_{tt}
\end{array}\right)\in\mathbb{F}_{q}^{s\times s}$ is a  block upper triangular, where each $A_{ij}$ is a $s_{i}\times s_{j}$ matrix  and $\sum_{i=1}^{t}s_{i}=s$. Then $r(A)\geq \sum_{i=1}^{t}r(A_{ii})$. The similar result of block lower triangular is also right.
\end{lem}

\begin{Theorem}
Let $C_{1},\cdots,C_{M}$ be linear codes of length $n$ over $\mathbb{F}_{q}$ and $C=[C_{1},\cdots, C_{M}]\cdot~A$, let $G_{i}$ be a  generator matrix of $C_{i}$ for all $1\leq i\leq M$, and $G$ be a generator matrix of $C$. Let $k_{i}$ be the  dimension of $C_{i}$ for all $1\leq i\leq M$, and $k$ be the dimension of $C$. Suppose $A\sigma^{\ell}(A^{T})=B=(b_{ij})_{M\times M}$ and $B$ is  upper triangular or  lower triangular.  Then

(1) If $b_{ii}\neq 0$ for all $1\leq i\leq M$ and  $A$ is right non-singular, then $$\sum_{i=1}^{M}r(G_{i}\sigma^{\ell}(G_{i})^{T})\leq r(G\sigma^{\ell}(G)^{T})\leq \sum_{i=1}^{M}k_{i}.$$

(2) If $r(G_{i}\sigma^{\ell}(G_{i})^{T})=k_{i}$, $A\sigma^{\ell}(A^{T})=\left(\begin{array}{cccc}
                       \lambda_{1}  &0  &\cdots & 0 \\
                        0  & \lambda_{2} &\cdots&\vdots \\
                        \vdots & \vdots &\ddots &0\\
                        0 & \cdots & 0 &\lambda_{M}
\end{array}\right)$ and $\lambda_{i}\neq0$ for all $1\leq i\leq M$. Then $\sum_{i=1}^{M}r(G_{i}\sigma^{\ell}(G_{i})^{T})= r(G\sigma^{\ell}(G)^{T})=\sum_{i=1}^{M}k_{i}.$
\end{Theorem}

\begin{proof}
(1) We just prove the case where $B$ is block upper triangular. Since $A=(a_{ij})$ is non-singular, we know that $G=\left(\begin{array}{ccc}
                       G_{1}a_{11}  &\cdots  &G_{1}a_{1N} \\
                        \vdots  & \ddots &\vdots \\
                        G_{M}a_{M1} & \cdots &G_{M}a_{MN}
\end{array}\right)$. Then we have

\begin{equation*}
\begin{split}
G\sigma^{\ell}(G)^{T}&=\left(\begin{array}{ccc}
                       G_{1}a_{11}  &\cdots  &G_{1}a_{1N} \\
                        \vdots  & \ddots &\vdots \\
                        G_{M}a_{M1} & \cdots &G_{M}a_{MN}
\end{array}\right)\left(\begin{array}{ccc}
                       \sigma^{\ell} (G_{1})^{T}\sigma^{\ell} (a_{11})  &\cdots  &\sigma^{\ell} (G_{M})^{T}\sigma^{\ell} (a_{M1}) \\
                        \vdots  & \ddots &\vdots \\
                        \sigma^{\ell} (G_{1})^{T}\sigma^{\ell} (a_{1N}) & \cdots &\sigma ^{\ell} (G_{M})^{T}\sigma^{\ell} (a_{MN})
\end{array}\right)\\
=&\left(\begin{array}{ccc}
                       b_{11}G_{1}\sigma^{\ell} (G_{1})^{T}  &\cdots  &b_{1M}G_{1}\sigma^{\ell} (G_{M})^{T} \\
                        \vdots  & \ddots &\vdots \\
                        b_{M1}G_{M}\sigma^{\ell} (G_{1})^{T} & \cdots &b_{MM}G_{M}\sigma ^{\ell} (G_{M})^{T}
\end{array}\right)=\left(\begin{array}{ccc}
                       b_{11}G_{1}\sigma^{\ell} (G_{1})^{T}  &\cdots  &b_{1M}G_{1}\sigma^{\ell} (G_{M})^{T} \\
                        \vdots  & \ddots &\vdots \\
                       0 & \cdots &b_{MM}G_{M}\sigma ^{\ell} (G_{M})^{T}
\end{array}\right).
\end{split}
\end{equation*}

By the above lemma, we know that $\sum_{i=1}^{M}r(G_{i}\sigma^{\ell}(G_{i})^{T})\leq r(G\sigma^{\ell}(G)^{T})$. Since $G\sigma^{\ell}(G)^{T}$ is $\sum_{i=1}^{M}k_{i}\times \sum_{i=1}^{M}k_{i}$, $r(G\sigma^{\ell}(G)^{T})\leq \sum_{i=1}^{M}k_{i}$. All in all,
$$
\sum_{i=1}^{M}r(G_{i}\sigma^{\ell}(G_{i})^{T})\leq r(G\sigma^{\ell}(G)^{T})\leq \sum_{i=1}^{M}k_{i}.
$$

(2) The proof of the second statement can be obtained from the proof of Statement~(1).
\end{proof}

\begin{Corollary}
Let $C_{1},\cdots,C_{M}$ be linear codes of length $n$ over $\mathbb{F}_{q}$ and $C=[C_{1},\cdots ,C_{M}]\cdot A$. Then

(1) If $A$ is right non-singular and $A\sigma^{\ell}(A^{T})$ is block upper triangular or block lower triangular, then $0\leq \dim(h_{\ell}(C))\leq \sum_{i=1}^{M}\dim(h_{\ell}(C_{i}))$.

(2) If  $A\sigma^{\ell}(A^{T})=diag(\lambda_{1},\cdots,\lambda_{M})$, where $\lambda_{i}\neq0$ for all $1\leq i\leq M$. Then $\dim(h_{\ell}(C))= \sum_{i=1}^{M}\dim(h_{\ell}(C_{i}))$.
\end{Corollary}

\begin{proof}
(1) Let $G_{i}$ be a generator matrix of $C_{i}$ for all $1\leq i\leq M$, $G$ be a  generator matrix of $C$.  Let $k_{i}$ be the dimension of $C_{i}$ for all $1\leq i\leq M$, and $k$ be the dimension of $C$. By the above theorem, we know that
$$
\sum_{i=1}^{M}r(G_{i}\sigma^{\ell}(G_{i})^{T})\leq r(G\sigma^{\ell}(G))\leq \sum_{i=1}^{M}k_{i}.
$$
By Theorem~2.1, we know that $k=r(G\sigma^{\ell}(G))+\dim(h_{\ell}(C))$ and $k_{i}=r(G_{i}\sigma^{\ell}(G_{i}))+\dim(h_{\ell}(C_{i}))$ for all $1\leq i\leq M$ . By Proposition 5.4, we know that $k=\sum_{i=1}^{M}k_{i}$. In summary, we  have $0\leq \dim(h_{\ell}(C))\leq \sum_{i=1}^{M}\dim(h_{\ell}(C_{i}))$.

(2) The poof of the second statement can be obtained from the proof of Statement~(1).
\end{proof}

\noindent {\bf Acknowledgement.} The authors sincerely thank Professor Jay A. Wood for his valuable comments during his visit at Central China Normal University from April to May in 2018. This work was supported by NSFC (Grant No. 11871025) and the self-determined research funds of CCNU from the colleges' basic research and operation of MOE (Grant No. CCNU18TS028).


\end{document}